\documentclass{article}
\usepackage{url}

\usepackage{amsmath}
\usepackage{times}
\usepackage{bm}
\usepackage{natbib}
\usepackage{color}
\usepackage{multirow}
\usepackage{graphicx}
\usepackage{varwidth}
\usepackage{float}
\usepackage[table,x11names]{xcolor}
\definecolor{light-gray}{gray}{0.95}
\usepackage{amssymb}
\usepackage{multirow}
\usepackage{geometry}
\usepackage{setspace}
\usepackage{booktabs}
\usepackage{amsthm}
\usepackage{diagbox}
\usepackage{hyperref}
\usepackage{MnSymbol}

\usepackage[skip=3pt]{caption}

\usepackage[most]{tcolorbox}
\newtcolorbox{mybox}[2][]{%
  attach boxed title to top center
               = {yshift=-8pt},
  colback      = blue!5!white,
  colframe     = blue!75!black,
  fonttitle    = \bfseries,
  colbacktitle = blue!85!black,
  title        = #2,#1,
  enhanced,
}   

\newtheorem{theorem}{Theorem}
\newtheorem{lemma}[theorem]{Lemma}

\DeclareMathOperator{\N}{\mathcal{N}}

\newcommand{\logit}{\textrm{logit}}

\title{Connecting Bayes factor and the Region of Practical Equivalence (ROPE) for testing interval null hypothesis }

\author{J. G. Liao, Vishal Midya, and Arthur Berg\\
 Division of Biostatistics and Bioinformatics\\
 Penn State University College of Medicine}
\date{}
\begin{document}

\maketitle
\begin{abstract} 

There has been strong recent interest in testing interval null  hypothesis  for improved scientific inference. For example, Lakens et al (2018) and Lakens and  Harms (2017) use this approach to study if there is a pre-specified meaningful treatment effect in gerontology and clinical trials, which is different from the more traditional point null hypothesis that tests for any treatment effect. Two popular Bayesian approaches are available for interval null hypothesis testing. One is the standard Bayes factor and the other is the Region of Practical Equivalence (ROPE) procedure championed by Kruschke and others over many years. This paper establishes a formal connection between these two approaches with two benefits. First, it helps to better understand and improve the ROPE procedure. Second, it leads to a simple and effective algorithm for computing Bayes factor in a wide range of problems using draws from posterior distributions generated by standard Bayesian programs such as BUGS, JAGS and Stan. The tedious and error-prone task of coding custom-made software specific for Bayes factor is then  avoided.
\end{abstract}


\section{Introduction}

Hypothesis testing is widely used in scientific research. For a statistical model parameterized in parameter $\theta$, the general formulation is: 
\[
H_0: \theta\in\Theta_0\quad\text{vs}\quad H_1:\theta\in\Theta_1.
\]
Traditionally, a point null hypothesis takes $\Theta_0$ to be a single point $\theta_0$ leading to the following 
\[
H_0: \theta=\theta_0\quad\text{vs}\quad H_1:\theta \ne \theta_0.
\]
However, many researchers, such as \cite{Meehl:1978aa} and \cite{Cohen:1994aa}, have argued that such formulation is not appropriate in most scientific researches because the parameter $\theta$ cannot be exactly specified to a single point $\theta_0$. For example when comparing a new treatment with a standard one, the differences in the mean of an outcome variable, say $\theta$, is almost never exactly 0. Therefore a more appropriate formulation is  
\[
\begin{split}
    H_0&:  \|\theta - \theta_0 \| \leq \delta\\
    H_1&:  \|\theta - \theta_0 \| > \delta,
\end{split}
\]
where $\theta$ that satisfies  $\|\theta - \theta_0\| \leq \delta$ represents practically negligible deviation from $\theta_0$. For example, \cite{Lakens:2018aa}  and  \cite{Lakens:2017}  use  this  approach to  study  if  there
is a pre-specified meaningful treatment effect in gerontology and clinical trials. \cite{Morey:2011aa} reviews and develops interval null hypothesis testing in the context of psychological research.

The frequentist approach to interval null hypothesis testing is to conduct an equivalence test \citep[e.g.][]{Lakens:2017,Rogers:1993,Wellek:2010}. The limitation of the frequentist hypothesis testing is well documented \citep{Wasserstein:2016aa}. In particular, it can only quantify evidence against the default hypothesis but not the evidence for it. Recently, there have been renewed interests \citep{Harms:2018ab,Lakens:2018aa,Morey:2011aa,Kruschke:2013aa} in using a Bayesian approach for tackling this interval null hypothesis problem so that $H_0$ and $H_1$ can be treated on a more equal basis. The standard Bayes approach is to use Bayes factor to quantify the relative support of the data for one hypothesis over the other. Standard references for Bayes factor include  \cite{Kass:1995aa,Berger:2013aa,Berger:1996aa,Berger:2001aa,Berger:2015aa}. The  special issue of \emph{Journal of Mathematical Psychology} (June, 2016) provides in-depth and updated discussion of Bayes Factor.

Alternatively, Krushke and others \citep{Kruschke:2018aa,Kruschke:2011aa,Kruschke:2013aa,Kruschke:2014aa, Carlin:2008aa,Edwards:2010aa,Freedman:1984aa,Hobbs:2007aa} have championed a procedure called ROPE (region of practical equivalence), in which the interval null is treated as a region of practical equivalence.  In this procedure, the 1-$\alpha$ highest density interval of the posterior distribution of $\theta$ is constructed. If this interval falls completely in $\Theta_j$, either $j=0$ or $j=1$, then $H_j$ is selected. Otherwise, the selection between $H_0$ and $H_1$ is declared uncertain.

The Bayes factor and ROPE have been treated as two different and distinctive procedures. In this paper, we provide a formal connection between them in Lemma \ref{lemma1} with two benefits.  First, it helps to better understand and improve the the ROPE procedure. Secondly and more importantly, it leads to a simple and effective algorithm for computing Bayes factor in a wide range of problems by using posterior draws generated by standard Bayesian software programs such as WinBUGS \citep{Lunn2000}, JAGS \citep{Plummer03jags:a}, and Stan \citep{STAN}. This circumvents the need for custom-made software to calculate marginal distributions for the Bayes factors. 



\section{Connection between Bayes factor and the posterior distribution}

We first outline the Bayes factor approach in testing $H_0$ vs. $H_1$. We need first to specify a prior of $\theta$ under $H_0$, $\pi_0$ on $\Theta_0$, and a prior of $\theta$ under $H_1$, $\pi_1$ on $\Theta_1$. Let $f(y\mid\theta)$ be the likelihood. The marginal distribution of $y$ under $H_j$ is then, 
\[
f(y\mid \pi_j) = \int {f(y\mid\theta) \pi_j (\theta) d\theta}, \quad \text{for } j=0,1.
\]
Here we use notation $f(y\mid \pi_j)$ to emphasize the dependence of marginal density on on prior $\pi_j$. Further assume that the prior probability is $\eta_1$ for for $H_1$ and $1 - \eta_1$ for $H_0$. It follows directly from Bayes Theorem that 
\begin{equation}
    \label{eq:BF}
\frac{\Pr(\theta \sim \pi_1\mid y)}{ 1- \Pr(\theta \sim \pi_1\mid y)} = \frac{f(y\mid \pi_1)}{f(y\mid \pi_0)} \frac{\eta_1}{1-\eta_1},
\end{equation}
where $\Pr(\theta \sim \pi_1\mid y)$ denotes the posterior probability that $\theta$ was generated from $\pi_1$.

The term $\frac{f(y\mid \pi_1)}{f(y\mid \pi_0)}$ is called the Bayes Factor, which quantifies the relative support in data for $\theta\sim \pi_1$ over $\theta\sim \pi_0$, independent of $\eta_1$. For interval null hypothesis, the Bayes factor and the posterior probability of $\theta$ are closely related as shown in Lemma \ref{lemma1}. 
\begin{lemma}
\label{lemma1}
Assume that $\Theta$ is the disjoint union of $\Theta_0$ and $\Theta_1$ and  $\pi_0(\theta)=0$ for $\theta\in \Theta_1$ and $\pi_1(\theta)=0$ for $\theta\in \Theta_0$. Let $\eta_1=\Pr(H_1)$ as defined before and define the combined prior of $\theta$ on $\Theta$ to be
\begin{equation}
\label{eq:mixture}
\pi= (1 - \eta_1)\pi_0+ \eta_1 \pi_1.
\end{equation}
Let $f(y\mid\theta)$ be the likelihood, and let $f(\theta\mid y)$ be the posterior density of $\theta$ under $\pi$.
Then 
\[
\Pr(\theta \sim \pi_1\mid y) = \int_{\Theta_1} f(\theta\mid y)d\theta
\]
and therefore
\[
\text{Bayes factor}\triangleq\frac{f(y\mid \pi_1)}{f(y\mid \pi_0)}=\frac{1-\eta_1}{\eta_1} \frac{\Pr(\theta \sim \pi_1 \mid y)}{1-\Pr(\theta \sim \pi_1 \mid y)}.
\]
\end{lemma}
\begin{proof}
\[
\begin{split}
\int_{\Theta_1} f(\theta\mid y)\, d\theta 
&= \frac{\int_{\Theta_1} \left \{ (1-\eta_1)\pi_0(\theta)+\eta_1\pi_1(\theta)  \right \} f(y\mid\theta)\,d\theta}{\int_{\Theta} \pi (\theta) f(y\mid\theta)\, d\theta}\\
&=\frac{\eta_1\int_{\Theta_1}\pi_1(\theta)f(y\mid\theta)\,d\theta}{\eta_1\int_{\Theta_1}\pi_1(\theta)f(y\mid\theta)\,d\theta + (1-\eta_1)\int_{\Theta_0}\pi_0(\theta)f(y\mid\theta)\,d\theta}\\
&=\Pr(\theta \sim \pi_1\mid y),
\end{split}
\]
which is the first expression of the Lemma. The second expression on Bayes factor follows from this result and \eqref{eq:BF}. 
\end{proof}
Although Lemma \ref{lemma1} is straightforward, it does not appear to have been published before. The two important implications of this Lemma are discussed in the next two sections.

\section{The ROPE Procedure}
In ROPE procedure, a prior distribution $\pi$ is directly specified for parameter $\theta$ on $\Theta$. The posterior distribution of $\theta$ is then computed. \cite{Kruschke:2018aa} summarized the ROPE procedure as follows. 

\begin{quote}
``Consider a ROPE around a null value of a parameter. If the $95\%$ HDI [\underline{H}ighest \underline{D}ensity \underline{I}nterval] of the parameter distribution falls completely outside the ROPE then reject the null value, because
the $95\%$ most credible values of the parameter are all not practically equivalent to the null value. If the $95\%$ HDI of the parameter distribution falls completely inside the ROPE then accept the null value for practical purposes, because the $95\%$ most credible values of the
parameter are all practically equivalent to the null value. Otherwise remain undecided, because some of the most credible values are practically equivalent to the null while other of the most credible values are not.''    
\end{quote}

 
The highest density interval (HDI) is used in ROPE because it is the $(1-\alpha)$ credible interval of shortest length/smallest volume.  This approach has a strong intuitive appeal and is similar to the frequentist equivalence testing (e.g., \cite{Lakens:2017}; \cite{Rogers:1993}; \cite{Wellek:2010}) in which a $(1-\alpha)$ confidence interval is used in place of the $(1-\alpha)$ highest density interval. 

We now provide additional details of the ROPE procedure. First, ROPE declares undecided between $H_0$ and $H_1$ when the $(1-\alpha)$ HDI is contained in neither $\Theta_0$ nor $\Theta_1$. This, however, is not fully informative as it fails to distinguish between two different cases: (1) the HDI is almost contained in $\Theta_0$ (or alternatively in $\Theta_1$); (2) the HDI is spread more evenly between $\Theta_0$ and $\Theta_1$. To improve on this, note that $\Pr(\theta \sim \pi_j \mid y) \geq (1-\alpha)$ if any $(1-\alpha)$ credible interval is contained in $\Theta_j$ for $j=0$ or $j=1$. Therefore, we propose to directly report $\Pr(\theta \sim \pi_j \mid y)$ with a decision rule to declare for $H_j$ when $\Pr(\theta \sim \pi_j \mid y) > 1-\alpha$. For the two cases above, $\Pr(\theta \sim  \pi_0 \mid y)$ is close to $1-\alpha$ for Case 1 but it is close to 0.5 for Case 2, which is more informative. In addition, $\Pr(\theta \sim \pi_j \mid y)$ connects ROPE directly to Bayes factor as seen in Lemma 1. Reporting $\Pr(\theta \in \Theta_j\mid y)=\Pr(\theta \sim \pi_j \mid y)$ was also proposed in  \cite{Wellek:2010} but  \cite{Kruschke:2018aa} seems to reject this idea in favor of using HDI.  Our analysis provides additional support for using $\Pr(\theta \sim \pi_j \mid y)$ instead of HDI. 

 
 
Second, Kruschke and others have touted the greater robustness of ROPE to prior $\pi$. Here we reformulate the ROPE procedure in a mathematically equivalent way so that it more closely resembles a Bayes factor formulation.  ROPE procedure starts with a single prior distribution, $\pi$, on $\Theta$, but by truncating this distribution on $\Theta_0$ and $\Theta_1$, respectively, we arrive at priors $\pi_0$ and $\pi_1$ along with $\eta_1$, the prior probability of $\Theta_0$. More specifically, we have 
\[
\begin{split}
& \pi_0(\theta) = \frac{\pi(\theta)}{\int_{\Theta_0} \pi(\theta) d\theta} \quad \text{for} \quad \theta \in \Theta_0\\    
& \pi_1(\theta) = \frac{\pi(\theta)}{\int_{\Theta_1} \pi(\theta) d\theta} \quad \text{for} \quad \theta \in \Theta_1\\    
& \eta_1 = \int_{\Theta_1} \pi(\theta) d\theta. \\
\end{split}
\]
This allows us to explore the greater robustness of the ROPE procedure over the Bayes factor approach. As an example, consider the hypotheses $H_0: |\theta|\le \delta$ vs $H_1: |\theta|>\delta$, and let $\pi \triangleq N(0,\tau^2)$.  As $\tau \to \infty$, $\pi_0$ approaches $\text{unif}(-\delta,\delta)$ and $\pi_1$ becomes more and more diffused on $(-\infty,-\delta)\cup(\delta,\infty)$. It is well known from Lindley's paradox that a very diffuse $\pi_1$ leads to increased support for $H_0$ over $H_1$ \citep{Robert:2014aa}. However at the same time, $\eta_1 \to 1$, which compensates for the more diffused $\pi_1$ as seen from Equation \ref{eq:BF}. This accounts for the apparent greater robustness of the ROPE procedure. Although the greater robustness is generally desirable, this implicit tangling of $\pi_0$, $\pi_1$ and $\eta_1$ can be difficult to comprehend for many statisticians. How to specify these quantities in a more meaningful way remains a topic for further research.


\section{Computing Bayes factors using MCMC draws from the posterior distribution}
Bayes factor usually requires custom software to compute the marginal distributions $f(y\mid \pi_1)$ and ${f(y\mid \pi_0)}$ through numerical integration. Coding such software can be a tedious and error-prone task. Lemma \ref{lemma1}, however, provides a simple way to compute Bayes factor using standard Bayesian programs for drawing from posterior distribution such as Stan \citep{STAN}, JAGS \citep{Plummer03jags:a}, and WinBUGS \citep{Lunn2000}. Let $\theta^{(1)},\ldots,\theta^{(N)}$ be draws from posterior distribution $p(\theta\mid y)$ under prior $\pi$ in Equation \eqref{eq:mixture}, where $N$ is the Monte Carlo sample size. Then under general conditions,
\[
\frac{1}{N}\sum_{i=1}^N 1_{\theta^{(i)}\in \Theta_1}\stackrel{N\rightarrow\infty}{\longrightarrow}\Pr(\theta \sim \pi_1\mid y)
\]
and
\begin{equation}
\label{eq:BFconvergence}
\frac{1-\eta_1}{\eta_1}\frac{\frac{1}{N}\sum_{i=1}^N 1_{\theta^{(i)}\in \Theta_1}}{1-\frac{1}{N}\sum_{i=1}^N 1_{\theta^{(i)}\in \Theta_1}}\stackrel{N\rightarrow\infty}{\longrightarrow}\text{BF}=\frac{f(y\mid \pi_1)}{f(y\mid \pi_0)}.
\end{equation}

There are several advantages of this approach of computing Bayes factor compared to custom software. First, it reduces programming cost considerably. Second,  packages such as Stan \citep{STAN}, JAGS \citep{Plummer03jags:a}, and WinBUGS \citep{Lunn2000} are well-tested, high quality software routines that are familiar to most statisticians performing Bayesian analyses. So it is more likely to get correct answers quickly. Third, people have gained considerable experience and understanding of the posterior distribution, and this approach facilitates the adoption of Bayes factor as an integrated part of Bayesian inference. The disadvantage is that it can be computationally less efficient than custom software specifically designed for Bayes factor.

Bayes factor does not depend on the value of $\eta_1$. In fact, \eqref{eq:BFconvergence} holds for for any $0<\eta_1<1$. However, based on our experience, the computational efficiency of drawing Monte Carlo samples from the posterior distribution can be greatly improved when $\eta_1$ is chosen so that the combined prior $\pi$ in Equation \eqref{eq:mixture} is made more continuous at the boundary between $\Theta_0$ and $\Theta_1$.

Note that Using MCMC draws to compute Bayes factor has been discussed previously. For example, \cite{Morey:2011ab} discusses using improved Savage-Dickey method to compute Bayes for nested models. \cite{Weinberg:2012aa} in a more general setting, proposes the use of harmonic mean approximation and other methods for computing the marginal distributions. Our method, by utilizing the special structure of interval hypotheses, is much simpler.  

We demonstrate our methods in the following two two examples.

\subsection{Example 1: Two sample t-test}
Here we consider a simple dataset from \cite{Lyle:1987aa}, which has also been analyzed in \cite{Gonen:2005aa} and \cite{Wang:2016aa}.  The data (reproduced in Table \ref{tab:calcium}) consists of changes in systolic blood pressure over 12 weeks for 21 African-American men, 10 of whom took calcium supplements and the remaining 11 took placebo supplements. Testing for equality of means between these two groups with a two-sample t-test (assuming equality of variances) yields a t-statistic of 1.63 and a p-value of 0.12. 
\begin{table}[h]
    \centering
\caption{Blood pressure change data from \cite{Lyle:1987aa}.}
\begin{tabular}{|c||ccccccccccc|}
\hline
     calcium & 7&-4&18&17&-3&-5&1&10&11&-2& \\
\hline     
     placebo & -1&12&-1&-3&3&-5&5&2&-11&-1&-3\\
\hline     
\end{tabular}
    \label{tab:calcium}
\end{table}

 We now proceed to analyze the data using the interval null hypothesis Bayes factor. Let $X_1,\ldots,X_{n}\stackrel{\text{iid}}{\sim}\N(\mu_X,\sigma^2)$, and let $Y_1,\ldots,Y_{m}\stackrel{\text{iid}}{\sim}\N(\mu_Y,\sigma^2)$.  We are interested in evaluating the standardized effect size 
\[
\theta=\frac{\mu_Y - \mu_X }{\sigma}.
\]
In particular, we consider the following hypotheses:
\[
\begin{split}
    H_0&:  |\theta| \leq \delta\\
    H_1&:  |\theta| > \delta,
\end{split}
\]
for some pre-specified $\delta>0$. To compute Bayes factor, we assume the following priors
\[
\begin{split}
& \mu_X\sim \N(0,100^2) \\
& \sigma^2\sim \text{Inverse-Gamma}(10^{-2},10^{-2})\\
& \theta\mid H_0 \triangleq \pi_0(\theta)\sim \text{Uniform}(-\delta,\delta)\\
&  \theta\mid H_1 \triangleq \pi_1(\theta)\sim N(0,\tau^2)I_{|\theta|>\delta},
\end{split}
\]
where $\tau$ in $\theta\in H_1$ is a parameter to be specified. Note that a larger $\tau$ corresponds to a prior distribution $\theta\in H_1$ that is more deviated from $\theta\mid H_0$. Let $\eta_0=\Pr(H_0)$ be the prior for hypothesis $H_0$. The combined prior $\pi$ for $\theta$ is then \eqref{eq:mixture} in Lemma \ref{lemma1}. 
Given the above priors for $\mu_A$, $\sigma^2$ and $\theta$ and the  likelihoods, we utilize Stan to generate draws $\theta^{(1)},\ldots,\theta^{(N)}$ from the posterior distribution of $\theta$. We then use Equation \eqref{eq:BFconvergence} for a consistent estimator of the Bayes factor.

We calculate Bayes factor of $H_1$ over $H_0$ for $\delta=0.1$, which is generally considered a small effect size, over a range of values for $\tau$ in the prior for $\theta\mid H_1$ that reflects different deviations of $H_1$ from $H_0$. The result is plotted in Figure \ref{fig:tau}. All these Bayes factors are between 1 and 1.4, indicating slightly more support of $H_1$ over $H_0$.

\begin{figure}[h]
    \centering
    \includegraphics[width=6cm]{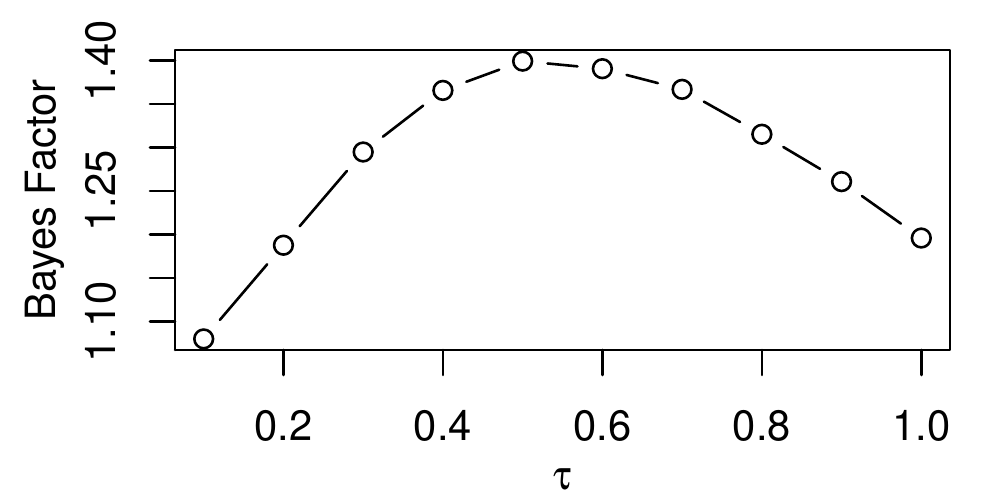}
    \caption{Bayes factors for the blood pressure data.}
    \label{fig:tau}
\end{figure}

It is noted that the Bayes factor for two-sample t-test setting with point null $H_0: \mu_X-\mu_Y = 0 $ has been extensively studied 
\citep{Rouder:2009aa,Wang:2016aa,Gonen:2005aa} including the R package \texttt{BayesFactor} by  \cite{Morey:2018aa}.

\subsection{Example 2: Meta analysis with a hierarchical Bayes model}

Here we will consider meta analysis of  a dataset originally published in Table 10 of \cite{Yusuf:1985}, which contains mortality data across 22 studies of patients who were treated with either a beta-blocker (treatment group) or placebo (control group) after experiencing a heart attack. The dataset is also reproduced and  analyzed in \cite[Sec. 5.6]{Gelman:2013ac}. The data from each study forms a $2\times 2$ table and the variables for the $i^\text{th}$ study are summarized in Table \ref{tab:2by2}.


\begin{table}[H]
    \centering
    \caption{Data format for the $i^\text{th}$ study}    
    \label{tab:2by2}
    \begin{tabular}{r c c l}
         & death & non-death & total \\\cline{2-3}
         \multicolumn{1}{r|}{beta-blocker} & \multicolumn{1}{c|}{$y_{i,1}$} & \multicolumn{1}{c|}{$n_{i,1}-y_{i,1}$} & $n_{i,1}$\\ \cline{2-3}
         \multicolumn{1}{r|}{placebo} & \multicolumn{1}{c|}{$y_{i,0}$} & \multicolumn{1}{c|}{$n_{i,0}-y_{i,0}$} & $n_{i,0}$\\\cline{2-3}
         
    \end{tabular}
\end{table}

For the $i^{\text{th}}$ study, let $p_{i,j}$, $j=0, 1$, be the underlying probability of death for the beta-blocker and placebo group respectively, and let $\theta_{i}$ be the corresponding log odds ratio; i.e. 
\[
\theta_i=\log \left(\frac{p_{i,1}}{1-p_{i,1}}\middle/
\frac{p_{i,0}}{1-p_{i,0}}\right)=\logit(p_{i,1})-\logit(p_{i,0}).
\]
We assume $\theta_i\sim\N(\theta,\sigma^2)$ and the focus of our inference is $\theta$, the mean of individual $\theta_i$. In particular, we compare  $H_0: |\theta| \leq \delta$ vs. $H_1 : |\theta|> \delta$. Our hierarchical Bayesian model is specified as follows:
\[
\begin{split}
y_{i,j}&\sim \text{Binomial}(n_{i,j},p_{i,j})\\
\logit(p_{i,0})&\sim \N(0,10^2)\\
\theta_i&\sim\N(\theta,\sigma^2)\\
\theta\mid H_0\triangleq \pi_0 &\sim \text{Uniform}[-\delta,\delta]\\
\theta\mid H_1\triangleq \pi_1 &\sim  N(0,\tau^2)I_{|\theta|>\delta}\\
\sigma^2 &\sim \text{Inverse-Gamma}(10^{-2},10^{-2}).
\end{split}
\]






For $\tau=1.5 \delta$ and  $\eta_0=1/2$, we draw MCMC samples from the posterior of $\theta$ using Stan and estimate $\Pr(\theta \in \pi_1\mid y)$ and $BF_{10}$ (Bayes Factor of $H_1$ over $H_0$) as in Lemma \ref{lemma1}. The results are given in Table 3 for three different values of effect size $\delta$. Note that here the Bayes factor shows strong support for $\pi_1$ over $\pi_0$ for $\delta=0.1$ whereas it shows a weak support of $\pi_1$ for $\delta=0.2$. When $\delta$ is increased to 0.3, however, Bayes factor shows weak support of $\pi_0$ over $\pi_1$. 

\begin{table}[H]
\centering
\caption{Bayes factors for different values of $\delta$}
\label{table:meta}
\begin{tabular}{|c|c|c|} 
\hline
 $\delta$ & $\Pr(\theta \in \pi_1\mid y)$ & $BF_{10}$ \\ 
\hline
$0.1$  & $0.97$ & $35.6$  \\ 
$0.2$ & $0.81$  & $4.37$ \\ 
$0.3$ & $0.29$  & $0.411$ \\ 
\hline
\end{tabular}
\end{table}

Stan and R code for Examples 1 and 2  are available at https://sites.google.com/site/jiangangliao.

In Summary, this papers establishes a formal connection between ROPE and Bayes factor for comparing interval hypotheses. This connection leads to better understanding of ROPE and provides a simple method to compute Bayes factor using MCMC draws from the posterior distribution.

\bibliographystyle{ECA_jasa}
\bibliography{BIC}

\end{document}